\documentclass[11pt,letterpaper]{article}

\usepackage[left=1in,right=1in,top=1in,bottom=0.9in]{geometry}

\usepackage{graphicx}
\usepackage[table,xcdraw]{xcolor}
\usepackage{wrapfig,booktabs}
\usepackage{setspace}
\usepackage{tikz}
\usepackage{amsthm}
\usetikzlibrary{decorations.shapes,shapes.geometric,shadows,arrows,automata,positioning,calendar,mindmap,backgrounds,scopes,chains,er,3d,matrix,fit}

\usepackage[utf8]{inputenc}
\usepackage{amsmath}
\usepackage{amsfonts,times}
\usepackage{cite}

\usepackage{lmodern}
\usepackage[most]{tcolorbox}

\usepackage{amssymb}
\usepackage{url}
\usepackage[hidelinks]{hyperref}

\usepackage[compact]{titlesec}
    \titlespacing{\section}{0pt}{3ex}{2ex}
    \titlespacing{\subsection}{0pt}{2.0ex}{1.0ex}
    \titlespacing{\subsubsection}{0pt}{1.0ex}{0.5ex}

\usepackage{array}
\newcolumntype{L}{>{\centering\arraybackslash}m{3cm}}

\newtheorem{theorem}{Theorem}
 \newtheorem{definition}[theorem]{Definition}
  \newtheorem{proposition}[theorem]{Proposition}

\begin{document}

\setstretch{1.67}

\pagenumbering{arabic}



\begin{center}
{\bf \Large Agentomics: Economic Foundations for the Valuation, Attribution, and Pricing of AI Agents in Human--AI Workflows \par}
\vspace{0.75em}
Quanyan Zhu\footnote{Department of Electrical and Computer Engineering, NYU Tandon School of Engineering, Brooklyn, NY 11201, USA. Email: \texttt{qz494@nyu.edu}.}
\end{center}

\begin{abstract}
Agentic AI systems are increasingly being deployed as productive resources in organizational workflows, yet existing evaluation methods primarily measure isolated technical performance rather than economic contribution. This paper introduces \emph{Agentomics}, a workflow-based framework for valuing, attributing, and pricing human and artificial agents. The framework models a workflow as a configuration of heterogeneous agents whose collective performance determines gross value, deployment cost, reliability, and expected failure loss. Workflow value is treated as a team-level quantity that may include complementarities, substitution effects, bottlenecks, and nonlinear production; additive stage-level value is only a special case. Building on this workflow model, the paper formulates AI deployment as a coalition-formation problem and defines coalition value as the incremental net surplus generated relative to a benchmark human workflow. The Shapley value is then used to attribute economic surplus among participating AI agents, yielding a principled connection among valuation, accountability, and market pricing. The resulting Shapley pricing equilibrium provides a normative benchmark for assessing whether agent prices reflect expected marginal contribution. A security-operations case study illustrates how the framework accounts for productivity gains, deployment costs, reliability losses, and coalition-level complementarities in hybrid human--AI workflows.
\end{abstract}



\section{Introduction}
\label{sec:introduction}

Artificial intelligence is rapidly evolving from a computational tool into an economic force. Recent advances in foundation models, large language models (LLMs), and agentic AI architectures have enabled AI systems to perform a growing range of cognitive tasks that were previously the exclusive domain of human workers. Modern AI agents can reason, plan, retrieve information, write software, monitor complex systems, generate forecasts, coordinate activities, and provide decision support with increasing levels of autonomy. As these capabilities continue to improve, organizations are beginning to deploy AI agents not merely as software applications but as productive resources that directly participate in organizational workflows and value creation.

Empirical evidence suggests that the economic implications of these developments are substantial. Brynjolfsson et al.~\cite{Brynjolfsson2023} report that a generative-AI assistant increased the productivity of 5,179 customer-support agents by approximately 13.8\% on average, with productivity gains exceeding 34\% among less experienced workers. Noy and Zhang~\cite{Noy2023} find that professionals using ChatGPT completed writing-intensive tasks approximately 40\% faster while simultaneously improving output quality by nearly 18\%. At a macroeconomic level, McKinsey estimates that generative AI technologies may contribute between \$2.6 trillion and \$4.4 trillion annually to the global economy~\cite{McKinsey2023}, while Goldman Sachs projects that AI-driven productivity growth could increase global GDP by approximately 7\% over the next decade~\cite{Goldman2023}. These findings suggest that AI agents are increasingly becoming economically significant productive assets whose deployment decisions may have consequences comparable to those associated with labor, capital, and technological innovation.

At the same time, organizations face increasingly difficult decisions regarding how AI agents should be integrated into production systems. Existing AI evaluation methodologies focus primarily on technical metrics such as benchmark accuracy, task-completion rates, coding performance, reasoning capability, and model size. While such metrics provide useful indicators of capability, they provide only limited insight into economic value. An agent that performs exceptionally well on technical benchmarks may create little organizational value if it introduces excessive supervision costs, coordination overhead, reliability risks, or downstream errors. Conversely, an agent with relatively modest benchmark performance may generate substantial economic value if it complements human expertise, reduces operational costs, improves workflow efficiency, or enhances organizational resilience. Consequently, technical capability and economic value are fundamentally distinct concepts.

This distinction becomes increasingly important as organizations adopt hybrid human--AI workflows. In practice, AI systems rarely replace entire organizational processes. Instead, AI adoption typically occurs incrementally, producing workflows in which humans and multiple AI agents interact across interconnected stages of decision making. Examples include cybersecurity operations that combine automated threat detection with human analysts, healthcare systems that integrate diagnostic agents with physician oversight, software-development pipelines involving coding agents and human engineers, and financial systems that combine AI-generated forecasts with human judgment. In such environments, workflow outcomes emerge from the collective behavior of multiple participants rather than from the actions of any single agent. This system-level perspective is consistent with game-theoretic approaches to cyber-physical, cyber-physical-human, and socio-technical systems, where security and resilience depend on interactions among heterogeneous agents across multiple layers~\cite{ZhuBasar2015,Huang2020DynamicGames,ZhuBasar2024SocioTechnical}. As a consequence, the economic contribution of an individual agent depends not only on its own capabilities but also on its interactions with other agents, human supervisors, upstream information sources, and downstream decision makers.

These observations give rise to a fundamental attribution problem. When a workflow succeeds, how should the resulting value be attributed among participating agents? When failures occur, how should responsibility be assigned? Which agents created value, which agents destroyed value, and which interactions contributed to the observed outcome? Existing evaluation methodologies provide limited support for answering these questions because they focus primarily on isolated task performance rather than workflow-level outcomes. Yet these questions lie at the heart of organizational decision making. Valuation requires identifying how much value an agent contributes. Pricing requires determining how much an agent should be compensated relative to that contribution. Accountability requires determining which agents should bear responsibility for success or failure. Governance and auditing require understanding how workflow outcomes emerge from interactions among heterogeneous participants.

The attribution problem becomes even more significant as AI systems become increasingly autonomous. Recent studies by the Model Evaluation and Threat Research (METR) organization suggest that the duration of tasks that frontier AI systems can complete with approximately 50\% reliability has been doubling roughly every seven months, indicating rapid improvements in long-horizon task capabilities~\cite{METR2025}. Nevertheless, reliability remains highly sensitive to workflow complexity, environmental uncertainty, and task duration. In domains such as healthcare, finance, cybersecurity, transportation, and critical infrastructure, even infrequent failures may generate substantial economic losses, safety risks, regulatory penalties, or reputational damage. These concerns are closely aligned with recent work on cyber resilience and strategic cybersecurity in the presence of intelligent, adaptive agents~\cite{Zhu2024CyberResilience,Zhu2025LLMAgenticCybersecurity}. Consequently, organizations require mechanisms not only for measuring the benefits of automation but also for quantifying the risks introduced by autonomous agents and assigning responsibility for their consequences.

These developments suggest that AI agents should be viewed as a new class of economic objects. Unlike traditional software systems, AI agents exhibit heterogeneous capabilities, uncertain performance, adaptive behavior, varying levels of autonomy, and measurable reliability profiles. Unlike human workers, AI agents can often be replicated at near-zero marginal deployment cost and can simultaneously participate in multiple workflows. They therefore occupy a unique position between labor and capital. Existing theories from labor economics, production economics, industrial organization, platform economics, and mechanism design provide important insights, yet none directly address the problems of agent-level valuation, attribution, accountability, and pricing in hybrid human--AI systems. This gap motivates the need for a dedicated economic framework for studying AI agents.

To address these challenges, we introduce \emph{Agentomics}, a framework for studying the economics, valuation, attribution, accountability, and pricing of human and artificial agents. Analogous to labor economics, which studies the allocation and valuation of human labor, Agentomics studies how heterogeneous agents create value, incur costs, generate risks, and interact within organizational workflows. The central premise of Agentomics is that the value of an agent should be determined by its contribution to workflow-level economic outcomes rather than by isolated measures of technical capability.

A key insight underlying Agentomics is that agent value is inherently relational rather than intrinsic. The contribution of an agent cannot generally be determined in isolation because value emerges from interactions among multiple participants. Consequently, AI deployment is modeled as a coalition-formation process in which subsets of AI agents are introduced into existing human-operated workflows. Different coalition structures generate different levels of value, cost, reliability, and operational risk. This perspective naturally leads to a cooperative-game-theoretic representation in which workflow outcomes are viewed as the result of interactions among participating agents.

Building upon this perspective, the paper develops a unified theory of valuation, attribution, accountability, and pricing. We first introduce a workflow-centric economic model that integrates productivity, deployment costs, workflow reliability, and failure losses into a unified measure of net workflow value. We then formulate hybrid human--AI workflows as cooperative games and define a coalition-value function that measures the incremental economic surplus generated by AI deployment relative to a benchmark human workflow. Agent value is subsequently defined through the Shapley value~\cite{Shapley1953}, which quantifies the expected marginal contribution of an agent across all possible deployment configurations. This formulation provides a rigorous mechanism for attributing both value creation and value destruction while accounting for complementarities, substitution effects, and workflow interactions.

The framework further establishes a direct connection among economic value, accountability, and market pricing. We introduce the notion of a \emph{Shapley Pricing Equilibrium}, under which the market price of an agent equals its expected marginal contribution to workflow surplus. This equilibrium concept provides a normative benchmark for evaluating whether agents are overvalued or undervalued relative to their economic contribution. More broadly, the same attribution mechanism that supports economic valuation also provides foundations for governance, auditing, responsibility attribution, and accountability in hybrid human--AI systems.

The contributions of this paper are fourfold. First, we introduce Agentomics as a new framework for studying AI agents as economic actors whose value is determined by their contribution to organizational outcomes. Second, we develop a workflow-centric theory of agent valuation that integrates productivity, deployment costs, reliability, and failure losses into a unified measure of economic surplus. Third, we formulate value attribution in hybrid human--AI systems as a cooperative game and derive Shapley-value-based measures of agent contribution. Fourth, we establish a direct connection among valuation, accountability, and market pricing through the notion of a Shapley Pricing Equilibrium, thereby providing a theoretical foundation for emerging markets in agentic AI services.

Collectively, these contributions establish a unified framework connecting workflow economics, reliability analysis, cooperative game theory, accountability, value attribution, and market pricing. Rather than evaluating AI systems solely according to benchmark accuracy or task-completion rates, Agentomics evaluates agents according to the economic surplus they generate within organizational workflows. As AI agents become increasingly integrated into organizational decision-making processes, understanding how such agents should be valued, deployed, governed, and priced will become a central challenge for researchers, practitioners, and policymakers alike.

The remainder of the paper is organized as follows. Section~\ref{sec:economic-foundations} develops the economic foundations of Agentomics. Section~\ref{sec:mathematical-framework} introduces the workflow-based Agentomics framework and formalizes the deployment of human and AI agents. Section~\ref{sec:agent-characteristics} develops the economic characterization of agents and derives measures of cost, performance, and value. Section~\ref{sec:workflow-economics} introduces workflow-level economics and derives measures of reliability, value, cost, and net economic surplus. Section~\ref{sec:coalitions} develops the coalition-based valuation framework and derives Shapley-based measures of agent value and pricing. Section~\ref{sec:case-study} illustrates the framework in a security-operations setting, and Section~\ref{sec:conclusion} concludes.

\section{Economic Foundations of Agentomics}
\label{sec:economic-foundations}

The emergence of foundation models, autonomous AI systems, and agentic workflows has created a new economic environment in which cognitive capabilities can be developed, purchased, allocated, and traded as productive resources. Unlike traditional software systems that automate predefined procedures, AI agents increasingly perform tasks involving perception, reasoning, planning, decision support, coordination, and autonomous action. As these systems become embedded within organizational workflows, they begin to function as economic actors whose deployment decisions influence productivity, costs, risks, and organizational outcomes.

The economic significance of AI agents is reflected in forecasts of rapid market expansion. Industry analyses project that the global market for AI agents may grow from approximately \$7.84 billion in 2025 to \$52.62 billion by 2030~\cite{MarketsandMarkets2025}. At the same time, advances in generative AI have demonstrated measurable productivity gains across a range of knowledge-intensive tasks. Brynjolfsson, Li, and Raymond report productivity improvements of approximately 14--15\% among customer-support workers using generative AI assistance, with gains exceeding 30\% for less experienced employees~\cite{Brynjolfsson2023}. Similarly, Noy and Zhang show that professionals using ChatGPT complete writing-intensive tasks approximately 40\% faster while improving output quality by nearly 18\%~\cite{Noy2023}. Together, these findings suggest that AI agents are evolving from experimental technologies into economically significant productive assets.

At the same time, organizations increasingly face decisions regarding which tasks should remain under human control and which should be delegated to artificial agents. These decisions depend not only on technical capability but also on deployment costs, operational risks, reliability, governance requirements, accountability obligations, latency, and scalability. Recent estimates suggest that generative AI technologies may ultimately contribute between \$2.6 trillion and \$4.4 trillion annually to the global economy~\cite{McKinsey2023}. Given the magnitude of these potential gains, even modest improvements in deployment efficiency, pricing, and allocation can have substantial economic consequences.

Traditional economic frameworks provide only a partial understanding of this phenomenon. Labor economics studies the allocation and valuation of human workers, while industrial organization examines firms, platforms, and markets. However, AI agents possess characteristics that distinguish them from both conventional software and human labor. They exhibit heterogeneous capabilities, varying levels of autonomy, explicit operational costs, measurable reliability profiles, and increasingly observable market prices. Furthermore, unlike human workers, AI agents can often be replicated at near-zero marginal deployment cost while simultaneously introducing new risks associated with reliability failures, governance requirements, and accountability concerns. These characteristics motivate the need for a dedicated economic framework for understanding the valuation, allocation, and pricing of AI agents.

\subsection{AI Agents as Economic Actors}

An AI agent is an autonomous or semi-autonomous software entity that perceives information, reasons over context, and takes actions in pursuit of specified objectives. Unlike traditional AI tools that merely provide recommendations or predictions, agents can interact with external systems, invoke tools, maintain memory, execute workflows, and adapt their behavior to changing environments. Modern agent architectures range from simple reactive agents to sophisticated planning systems and multi-agent organizations capable of coordinating complex tasks.

From an economic perspective, AI agents may be viewed as productive assets. They consume resources, generate outputs, create value, and impose costs. Like human workers, agents possess heterogeneous capabilities that influence their effectiveness across different tasks. Some agents specialize in narrow domains such as customer support, legal document review, software development, or cybersecurity monitoring, while others operate as general-purpose assistants capable of supporting a broad range of functions. Consequently, agent deployment can be viewed as a resource-allocation problem in which organizations seek to match capabilities to tasks in order to maximize value creation.

A distinguishing characteristic of AI agents is that they increasingly participate in explicit markets. Organizations can purchase access to agent capabilities through subscriptions, APIs, cloud marketplaces, licensing agreements, and dedicated agent stores. Examples include OpenAI's GPT Store, Anthropic's Claude Skills ecosystem, and cloud marketplaces that support AI-agent listings~\cite{OpenAIStore2024,AnthropicSkills2026,GoogleMarketplace2026}. The existence of observable market prices transforms agent deployment into an economic decision in which organizations compare the expected value generated by an agent against its acquisition and operational costs.

The emergence of agent marketplaces further suggests that AI agents are evolving into economic goods. Similar to labor markets, where wages reflect the perceived value of human capabilities, agent markets provide mechanisms through which the value of computational capabilities can be expressed through prices. However, unlike labor markets, current pricing mechanisms remain largely heuristic, often based on subscriptions, token consumption, or platform-specific pricing policies. A central objective of Agentomics is therefore to establish a principled connection between agent capabilities, economic value, and market prices.

\subsection{Cost Structure of AI Agents}

A common misconception is that the economics of AI agents are primarily determined by model inference costs. In practice, agent deployment involves a multi-layered cost structure extending far beyond API charges and token consumption. Total cost of ownership includes development expenditures, computational resources, data acquisition and maintenance, infrastructure provisioning, monitoring, governance, compliance, security, and human oversight.

Development costs encompass requirements analysis, prompt engineering, workflow design, software integration, testing, and deployment. Depending on complexity, reported development costs range from several thousand dollars for simple agents to hundreds of thousands of dollars for enterprise-scale systems. These costs are often amplified in domain-specific applications requiring specialized knowledge, proprietary data, or regulatory compliance.

Operational costs include cloud computing resources, GPU utilization, storage, networking, and API expenditures. Although token-based billing receives significant attention, deployed agentic systems also require integration, monitoring, maintenance, governance, and human support. Consequently, evaluating agent economics solely through inference costs can substantially underestimate deployment expenses.

Data-related costs constitute another important component of agent economics. Fine-tuning, retrieval-augmented generation systems, specialized knowledge bases, and continuous model adaptation require investments in data acquisition, curation, storage, and maintenance. In many enterprise applications, data preparation and governance represent some of the largest expenditures associated with deployment.

Maintenance and monitoring costs persist throughout the agent lifecycle. Deployed systems require performance monitoring, logging, model updates, security reviews, bug fixes, and adaptation to evolving user requirements. Surveys indicate that many production systems continue to rely on human oversight loops, creating additional labor costs that persist even after automation is introduced.

An additional source of cost arises from reliability limitations. Agentic workflows frequently involve multiple sequential decisions in which errors accumulate across stages. Recent studies by METR demonstrate that while frontier AI systems can successfully complete increasingly long tasks, reliability remains highly sensitive to task complexity and workflow length~\cite{METR2025}. Even small stage-level failure probabilities can generate substantial workflow-level risks when decisions are chained together. The resulting need for verification, re-execution, exception handling, and human intervention creates what practitioners increasingly describe as an \emph{unreliability tax}. This observation implies that the economic attractiveness of an agent depends not only on its average performance but also on the costs associated with failures. These considerations motivate the explicit incorporation of reliability and failure losses into the Agentomics framework developed later in this paper.

\subsection{Agent Marketplaces and Pricing Models}

The commercialization of AI agents has produced a diverse collection of pricing mechanisms. Current markets employ subscription pricing, usage-based pricing, outcome-based compensation, revenue-sharing agreements, and hybrid models that combine multiple approaches. The diversity of these mechanisms reflects the absence of a universally accepted theory of agent valuation.

Subscription models charge fixed monthly or annual fees for access to agent services. Usage-based models charge according to API calls, token consumption, interactions, or completed tasks. Outcome-based models tie compensation directly to realized value, such as generated sales, cost reductions, or productivity improvements. Hybrid models increasingly combine subscriptions with usage-based overage charges.

The emergence of agent marketplaces has further complicated pricing decisions. Platforms must determine how to distribute revenue among infrastructure providers, marketplace operators, and agent developers. At the same time, developers must choose among competing marketplaces that differ in fees, visibility, distribution reach, and governance policies. These interactions create a two-sided market in which platform design influences both agent supply and user demand. More generally, distributed agentic workflows raise incentive-compatibility questions because agents may be supplied by different providers, participate in different teams, and require mechanisms that align local incentives with workflow-level objectives~\cite{YangZhu2026InternetAgenticAI}.

Economic theory suggests that agent pricing should ultimately depend on value creation. However, most observed pricing mechanisms rely on subscriptions, token usage, or competitive benchmarking rather than rigorous measures of economic contribution. Consequently, a fundamental question remains unanswered: what should determine the price of an AI agent? This paper addresses that question by linking pricing to the economic surplus generated by agents within workflows. The resulting framework provides a normative benchmark against which observed market prices can be evaluated.

\subsection{Human--AI Substitution and Complementarity}

Public discussions of AI frequently focus on labor substitution. While automation undoubtedly replaces certain tasks, many real-world deployments exhibit a more complex relationship between human and artificial agents. In practice, organizations increasingly employ hybrid human--AI workflows in which humans and agents collaborate to achieve shared objectives. Related work on cyber-physical-human and socio-technical systems emphasizes that system performance depends on coordinated decisions across human, cyber, and physical layers rather than on isolated component behavior~\cite{Huang2020DynamicGames,ZhuBasar2024SocioTechnical}.

Human operators contribute contextual understanding, adaptability, creativity, ethical judgment, and resilience under uncertainty. AI agents contribute scalability, speed, consistency, and low marginal execution costs. The resulting relationship is frequently complementary rather than purely substitutive. For example, cybersecurity analysts increasingly rely on AI agents for threat triage and anomaly detection while retaining responsibility for strategic decisions and incident response. Similarly, physicians may use diagnostic agents to augment decision-making without relinquishing clinical judgment.

These complementarities imply that the value of an AI agent cannot be evaluated independently of its operating environment. An agent's contribution depends not only on its own capabilities but also on its interactions with human operators, upstream information sources, downstream decision-makers, and other AI agents. Consequently, evaluating agents in isolation can produce misleading conclusions regarding their economic value.

This observation motivates the coalition-based perspective adopted in the remainder of the paper. Rather than viewing AI deployment as a binary choice between humans and machines, Agentomics models organizations as collections of interacting human and artificial agents whose joint actions determine workflow outcomes. The economic contribution of any individual agent therefore depends on the broader coalition within which it operates. This perspective naturally leads to cooperative-game-theoretic methods for measuring value creation, attributing contributions, and establishing economically meaningful prices for AI agents.

\subsection{Policy, Governance, and Accountability}

The increasing deployment of autonomous agents raises important questions regarding governance, accountability, and market regulation. As agents gain the ability to negotiate, transact, recommend, and make decisions on behalf of users, traditional notions of responsibility become more difficult to apply. Determining who is accountable when an agent produces harmful outcomes remains a significant challenge for organizations and policymakers. This issue is closely related to strategic trust, cyber-deception, and agentic workflow design models, which show that trust, manipulation, detection, and responsibility must be analyzed as interaction-dependent properties rather than as static attributes of a single participant~\cite{PawlickZhu2017StrategicTrust,GeZhu2024TrustAI,PawlickColbertZhu2019DeceptionSurvey,HuangZhu2021Duplicity,AlBariZhu2025Gestalt,Zhu2025GenerativeConjectural}.

Emerging regulatory frameworks, including the European Union's AI Act and Digital Markets Act, signal growing recognition of these issues~\cite{EUAIAct2024,DMA2022}. These initiatives emphasize transparency, auditability, human oversight, interoperability, and risk management for AI-enabled systems. Simultaneously, competition authorities have expressed concern that foundation-model ecosystems may generate new forms of market concentration and platform gatekeeping~\cite{CMA2023,CMA2024}.

These governance concerns are closely related to the economic questions addressed in this paper. Pricing, accountability, and value attribution are fundamentally connected. If organizations are to rely upon autonomous agents in economically significant decisions, they must possess principled mechanisms for quantifying both the value created by agents and the responsibility associated with their actions. The Agentomics framework developed in the remainder of this paper seeks to provide part of this foundation by connecting workflow economics, reliability analysis, cooperative-game-theoretic attribution, and market pricing within a unified analytical framework.

\section{Mathematical Framework of Agentomics}
\label{sec:mathematical-framework}

The emergence of foundation models, autonomous AI systems, and specialized reasoning agents has created a new economic environment in which cognitive capabilities can be purchased, allocated, and substituted across organizational workflows. As AI agents increasingly perform tasks traditionally carried out by human workers, firms face a new class of resource-allocation problems concerning the deployment of heterogeneous human and artificial agents. These decisions are driven not only by technical performance but also by economic considerations such as cost, reliability, latency, scalability, and operational risk. We refer to the study of these interactions as \emph{Agentomics}, namely the economics of AI agents. Analogous to labor economics, which studies the allocation and valuation of human labor, Agentomics seeks to understand the pricing, allocation, productivity, and substitution of AI agents within production systems.

Consider a workflow consisting of \(K\) ordered stages indexed by \(\mathcal K=\{1,\ldots,K\}\). Throughout the paper, the index \(k\) denotes a stage, and \(\tau_k\) denotes the task executed at stage \(k\). Each stage represents a productive activity that transforms information, generates decisions, or performs actions contributing to a broader organizational objective. The output of stage \(k\) serves as the input to stage \(k+1\), thereby creating a sequence of interdependent decision processes. Such workflows arise naturally in healthcare, cybersecurity, finance, manufacturing, and software engineering. For example, a medical diagnostic workflow may consist of imaging analysis, symptom interpretation, diagnosis generation, and treatment recommendation. Because stages are interconnected, errors, uncertainty, or delays introduced upstream may propagate downstream and affect overall workflow performance. Consequently, organizational outcomes depend not only on the quality of individual agents but also on how agents interact through the workflow structure.

Each stage \(k\in\mathcal K\) is associated with a task \(\tau_k\in\mathcal T\) and a task description \(d_k \in \mathcal D\), where \(\mathcal T\) denotes the task space and \(\mathcal D\) denotes the space of task specifications. The object \(d_k\) represents the requirements of task \(\tau_k\) at stage \(k\). It need not be a numerical vector; it may be a set, a weighted set, a tuple of attributes, or another structured specification. A useful formal interpretation is to let \(\Omega\) denote a universe of requirement and capability attributes, and to view \(d_k\subseteq\Omega\) as the set of attributes required by stage \(k\). These attributes may include domain expertise, contextual reasoning, perception, planning, verification burden, latency tolerance, autonomy requirement, security sensitivity, or consequence severity. Thus \(d_k\) is not merely a text label for the task; it is a formal requirement profile for the stage. For example, in a SOC workflow, alert triage may require log analysis, threat prioritization, low latency, and uncertainty handling, while response recommendation may require causal reasoning, accountability, policy compliance, and human-verifiable justification. The collection \((d_1,\ldots,d_K)\) defines the cognitive architecture of the workflow and determines the capabilities required at each stage. Since stage tasks may differ substantially in complexity, uncertainty, and expertise requirements, the effectiveness of a particular agent depends critically on the degree to which the agent's capabilities cover the requirements of the task \(\tau_k\) assigned to that stage.

The organization has access to a set of agents \(\mathcal A\). We partition the agent space into human and artificial agents such that \(\mathcal A=\mathcal H \cup \mathcal I\) and \(\mathcal H\cap\mathcal I=\emptyset\), where \(\mathcal H\) denotes human agents and \(\mathcal I\) denotes AI agents. Human agents represent the traditional workforce, while AI agents correspond to foundation models, specialized machine-learning systems, autonomous software agents, or other computational decision-making entities. From the perspective of Agentomics, both classes of agents constitute productive resources that can be allocated across workflow stages to generate value.

Each agent \(a \in \mathcal A\) possesses a capability profile \(m_a \in \mathcal M\), where \(\mathcal M\) denotes the capability-profile space. The object \(m_a\) describes what agent \(a\) can provide, and it may also be represented as a subset \(m_a\subseteq\Omega\), a weighted capability set, or a structured profile. For AI agents, \(m_a\) may include model accuracy, domain specialization, tool-use ability, robustness under distribution shift, interpretability, low-latency execution, autonomy, and auditability. For human agents, it may include expertise, judgment, adaptability, institutional knowledge, accountability, and supervisory capacity. Thus \(d_k\) describes what stage \(k\) requires, while \(m_a\) describes what agent \(a\) can provide. As an illustration, a triage agent may have \(m_a=\{\text{log analysis},\text{threat scoring},\text{low latency}\}\), which covers much of \(d_2\) but may cover less of \(d_4\) if the response-recommendation stage requires accountability, policy judgment, and human-verifiable justification. Comparing \(d_k\) and \(m_a\) provides a mechanism for matching agents to workflow requirements and determines the productive contribution of an agent when assigned to stage \(k\).

A distinguishing feature of Agentomics is that AI agents possess explicit market prices. Unlike human labor, which is typically acquired through employment contracts, AI capabilities are increasingly supplied through APIs, subscriptions, licensing agreements, and agent marketplaces. Consequently, each agent is associated not only with a capability profile but also with an economic profile. For agent \(a\), let \(p_a\) denote its market price, \(\ell_a\) its execution latency, and \(r_a\) its operational risk. Stage-specific deployment costs are introduced below as \(c_k(a)\). For human agents, costs may reflect wages, benefits, training expenditures, and opportunity costs. For AI agents, costs may include inference expenses, API charges, licensing fees, computational resources, monitoring costs, and governance overhead. The coexistence of capability and economic attributes transforms agent selection into a resource-allocation problem in which organizations must balance performance against cost and risk.

A workflow configuration is an assignment \(x=(x_1,\ldots,x_K)\in\mathcal X\), where \(x_k \in \mathcal A_k\subseteq\mathcal A\) denotes the agent assigned to stage \(k\), \(\mathcal A_k\) is the feasible agent set for that stage, and \(\mathcal X\subseteq\prod_{k=1}^K\mathcal A_k\) denotes the feasible configuration set. Thus, in the baseline formulation, each stage has one responsible agent selected from its feasible set. The feasible set \(\mathcal A_k\) may contain multiple human candidates and multiple AI candidates capable of performing \(\tau_k\); the framework therefore allows two or more AI agents to be feasible substitutes for a single stage. A configuration therefore specifies a complete realization of the workflow and determines how productive responsibilities are distributed across human and AI resources. Different configurations induce different performance levels, cost structures, reliability characteristics, and patterns of error propagation. The benchmark human workflow is represented by \(x^H=(H_1,\ldots,H_K)\), where \(H_k \in \mathcal H\cap\mathcal A_k\) for all \(k\). The superscript \(H\) in \(x^H\) denotes the all-human benchmark configuration, whereas \(\mathcal H\) denotes the set of human agents. This benchmark corresponds to the incumbent production system and serves as the reference point against which AI-enabled workflows are evaluated.

More generally, configurations may be hybrid, combining both human and AI agents across different stages. Such hybrid arrangements are particularly important because organizations typically adopt AI incrementally rather than replacing entire workflows simultaneously. The resulting design problem is therefore one of determining which stages should remain human-operated and which stages should be delegated to AI systems for execution of their associated tasks. Under the Agentomics perspective, workflow design becomes an economic optimization problem in which the objective is not merely to maximize technical performance but rather to optimize the joint trade-off among productivity, cost, risk, latency, and scalability. The fundamental question is therefore how heterogeneous human and artificial agents should be allocated across a workflow in order to maximize organizational value given their capabilities, prices, and operational characteristics. This question forms the foundation of Agentomics and motivates the analytical models developed in the remainder of the paper.

\section{Economic Characteristics of Agents}
\label{sec:agent-characteristics}

The central problem of Agentomics is to determine how heterogeneous human and artificial agents should be allocated across workflow stages when agents differ both in their capabilities and in their economic characteristics. Once the workflow structure and feasible stage-agent assignments have been specified, the desirability of a workflow configuration depends on two fundamental considerations. The first is the economic cost associated with acquiring and deploying agents throughout the workflow. The second is the performance generated by those agents when executing the task associated with each assigned stage. Together, these quantities determine the value created by a workflow configuration and provide the basis for economic comparison among alternative allocations of human and AI resources.

A distinguishing feature of modern AI systems is that they are increasingly traded as market goods. Access to foundation models, reasoning agents, and specialized AI services is commonly obtained through subscription contracts, usage-based pricing schemes, enterprise licenses, or agent marketplaces. Consequently, AI agents possess explicit prices that can be compared directly against the costs of human labor. Agentomics seeks to understand how these prices interact with agent capabilities to influence organizational decisions regarding automation, workforce composition, and workflow design.

\subsection{Agent Pricing and Deployment Costs}

Each agent \(a \in \mathcal A\) is associated with a market price \(p_a \geq 0\). For human agents, \(p_a\) may represent wages, compensation, benefits, training expenditures, and opportunity costs. For AI agents, \(p_a\) may represent licensing fees, subscription charges, API costs, or other expenditures required to access and utilize the agent.

In addition to acquisition costs, organizations incur operational expenses when deploying agents within workflows. Let \(o_k(a)\) denote the operational cost associated with assigning agent \(a\) to stage \(k\), whose task is \(\tau_k\). Depending on the application domain, operational costs may include computational expenditure, infrastructure utilization, monitoring overhead, supervisory effort, verification requirements, communication costs, energy consumption, or latency-related penalties. The total economic cost of assigning agent \(a\) to stage \(k\) is therefore represented by \(c_k(a)=p_a+o_k(a)\), where \(c_k(a)\ge 0\).

More generally, operational costs depend not only on the characteristics of the agent but also on the degree to which agent capabilities cover the requirements of the task executed at a stage. Agents that are poorly matched to a stage often require additional supervision, correction, or re-execution, thereby increasing deployment costs. Let \(\mu(d_k,m_a)\in[0,1]\) denote a coverage measure between the task-description object \(d_k\) for stage \(k\) and the capability profile \(m_a\) of agent \(a\). A larger value of \(\mu(d_k,m_a)\) indicates that the agent covers a larger portion of the stage requirements. If \(d_k,m_a\subseteq\Omega\) are sets of required and supplied attributes, a canonical specification is \(\mu(d_k,m_a)=\nu(d_k\cap m_a)/\nu(d_k)\), where \(\nu\) is a nonnegative measure or importance weight on \(\Omega\) and \(\nu(d_k)>0\). This value is one when the agent covers all required attributes and zero when it covers none. Weighted numerical vectors are a special case: if \(d_k,m_a\in[0,1]^J\), then \(\mu(d_k,m_a)=\sum_{j=1}^{J}w_j\min\{d_{kj},m_{aj}\}/\sum_{j=1}^{J}w_jd_{kj}\). To capture this dependence, we define
\begin{equation}
\label{eq:stage-cost}
c_k(a)
=
c\!\left(a,\tau_k,\mu(d_k,m_a)\right),
\end{equation}
\noindent where \(c:\mathcal A\times\mathcal T\times[0,1]\rightarrow\mathbb R_+\) is a stage-cost model and \(\mu:\mathcal D\times\mathcal M\rightarrow[0,1]\) is the coverage measure from agent capabilities to task requirements. 
The collection of stage-dependent cost functions \(\{c_k\}_{k=1}^{K}\) characterizes the economic burden associated with executing a workflow under a given allocation of agents.

\subsection{Agent Performance}

Agent performance is modeled through a general map from agents, stage tasks, and coverage levels to success probabilities:
\begin{equation}
\label{eq:performance-map}
\rho :
\mathcal A
\times
\mathcal T
\times
[0,1]
\rightarrow
[0,1]
\end{equation}
\noindent For agent \(a\), task \(\tau_k\), and coverage level \(\mu(d_k,m_a)\), the quantity \(\rho(a,\tau_k,\mu(d_k,m_a))\) gives the probability that agent \(a\) successfully executes task \(\tau_k\) at stage \(k\). This map permits performance to depend on intrinsic agent characteristics, task identity, coverage of task requirements, contextual uncertainty, environmental conditions, domain-specific execution models, and other operational factors.

The stage-specific performance function \(\rho_k:\mathcal A_k\rightarrow[0,1]\) is obtained by evaluating the general performance map at the task associated with stage \(k\):
\begin{equation}
\label{eq:stage-performance-from-map}
\rho_k(a)
=
\rho\!\left(a,\tau_k,\mu(d_k,m_a)\right).
\end{equation}
\noindent Thus \(\rho\) is the general performance model, whereas \(\rho_k(a)\) is the success probability of assigning agent \(a\) to the particular stage \(k\). If \(E_k(a)\) denotes the event that task \(\tau_k\) is completed correctly by agent \(a\), then
\begin{equation}
\label{eq:stage-performance}
\rho_k(a)
=
\mathbb P\!\bigl(E_k(a)\bigr).
\end{equation}

\noindent The quantity \(\rho_k(a)\) measures the effectiveness of agent \(a\) when performing task \(\tau_k\) at stage \(k\). Successful execution means that the output generated at stage \(k\) satisfies the semantic and operational requirements necessary for downstream workflow processing.

Performance depends jointly on intrinsic agent quality and compatibility with the task \(\tau_k\) associated with stage \(k\). Even highly capable agents may exhibit degraded performance when deployed outside their intended domain, whereas strong alignment between task requirements and agent capabilities can significantly improve outcomes. In healthcare applications, for example, an imaging-analysis model may achieve high performance at an image-analysis stage while exhibiting substantially lower performance at diagnosis-generation or treatment-planning stages that require contextual reasoning and clinical judgment.

A simple and interpretable specification is
\begin{equation}
\label{eq:compatibility-performance}
\rho_k(a)
=
\rho_a^{0}\,\mu(d_k,m_a),
\end{equation}
\noindent where \(\rho_a^{0}\in[0,1]\) denotes the baseline performance level of agent \(a\). Under this special model, the general map satisfies \(\rho(a,\tau_k,z)=\rho_a^{0}z\) for coverage level \(z\). Hence \(\rho_a^{0}\) represents the probability that the agent successfully performs a task for which it is fully qualified, while the coverage measure \(\mu(d_k,m_a)\) adjusts performance according to the requirements of task \(\tau_k\) at stage \(k\). This multiplicative specification is only one convenient model; the general map \(\rho\) may be nonlinear, task-dependent, context-dependent, or estimated empirically.

\subsection{Agent Value}

Cost and performance are not ends in themselves; rather, they jointly determine the economic value generated by an agent. The central object of Agentomics is therefore the relationship between the value created by an agent and the resources required to deploy it.

The local expected value generated at stage \(k\) is represented by a stage-value function \(v_k:\mathcal A_k\rightarrow\mathbb R\), where \(v_k(a)\) denotes the value generated when agent \(a\) executes task \(\tau_k\). Depending on the application, value may correspond to revenue generation, productivity gains, quality improvements, risk reduction, mission effectiveness, or other organizational objectives. The function \(v_k\) may already incorporate the success probability \(\rho_k(a)\), the quality of the produced output, and any stage-local benefits or penalties. It is used only as a local or additive special case; the general workflow-value function introduced later need not decompose into stage-level values.

The economic attractiveness of an agent is determined by comparing generated value with incurred cost. A useful measure of economic efficiency is

\begin{equation}
\label{eq:stage-efficiency}
\eta_k(a)
=
\frac{v_k(a)}{c_k(a)},
\end{equation}

\noindent which is well defined whenever \(c_k(a)>0\) and may be interpreted as the expected value generated per unit cost. Agents with larger values of \(\eta_k(a)\) provide greater economic return on deployed resources, although this ratio should be interpreted together with workflow-level reliability and downstream interaction effects rather than as a complete deployment criterion.

This perspective highlights a fundamental distinction between technical performance and economic desirability. An agent with the highest reliability is not necessarily the economically optimal choice if it is substantially more expensive than competing alternatives. Conversely, a lower-cost agent may be preferred if the resulting reduction in performance is outweighed by the associated cost savings.

\subsection{Human--AI Economic Tradeoffs}

The preceding framework captures the economic tradeoffs that arise when allocating human and AI agents across workflows. Human agents often possess superior contextual reasoning, adaptability, creativity, and robustness under ambiguity. These attributes are particularly valuable in environments characterized by incomplete information, novel situations, or complex judgment. However, human labor is frequently associated with higher costs, limited scalability, and slower execution speeds.

AI agents, by contrast, offer rapid inference, low marginal execution costs, and large-scale automation. These characteristics make them attractive for repetitive, high-volume, and computationally intensive tasks. Nevertheless, AI systems may exhibit reduced robustness under distributional shift, rare edge cases, adversarial conditions, or semantically mismatched assignments.

As a consequence, neither class of agents is universally superior. The relative economic value of an agent depends on the interaction among its capabilities, its price, the requirements of the assigned task, and its contribution to downstream workflow outcomes. Agentomics therefore evaluates agents not in isolation but according to their impact on end-to-end organizational performance.

The framework shifts the focus from component-level benchmarking to system-level economic optimization. The objective is not merely to identify the most capable agents, but rather to determine the allocation of human and artificial agents that maximizes organizational value while accounting for costs, risks, performance, and operational constraints. This allocation problem constitutes the central optimization challenge studied in Agentomics.

\section{Workflow Economics}
\label{sec:workflow-economics}

Having defined the economic and performance characteristics of individual agents, we now evaluate entire workflow configurations. Because workflow stages are sequentially interconnected, the economic performance of a workflow depends not only on the capabilities and costs of individual agents but also on how their contributions interact throughout the production process. Agentomics therefore evaluates workflow configurations using system-level measures of generated value, deployment cost, and execution risk. Throughout this section, let \(x=(x_1,\ldots,x_K)\in\mathcal X\) denote a feasible workflow configuration as defined in Section~\ref{sec:mathematical-framework}.

\subsection{Workflow Reliability}

For each stage \(k\), let \(E_k(x_k)\) denote the event that task \(\tau_k\) is successfully executed by the assigned agent \(x_k\). Since the workflow is sequential, successful completion of the entire workflow requires successful execution at every stage. The workflow success event is therefore

\begin{equation}
\label{eq:workflow-success-event}
E(x)
=
\bigcap_{k=1}^{K} E_k(x_k).
\end{equation}

\noindent Assuming stage-execution outcomes are conditionally independent given the workflow configuration \(x\), the reliability of the workflow is

\begin{equation}
\label{eq:workflow-reliability}
R(x)
=
\mathbb P(E(x))
=
\prod_{k=1}^{K}\rho_k(x_k),
\end{equation}

\noindent where \(\rho_k(x_k)\) is the stage-specific success probability induced by the general performance map \(\rho\) in Section~\ref{sec:agent-characteristics}; equivalently, \(\rho_k(x_k)=\rho(x_k,\tau_k,\mu(d_k,m_{x_k}))\). The corresponding workflow failure probability is

\begin{equation}
\label{eq:workflow-failure-probability}
q(x)
=
1-R(x).
\end{equation}

\noindent The multiplicative structure reflects the cumulative nature of reliability in sequential production systems. Since successful workflow completion requires success at every stage, a single low-performance assignment can substantially reduce end-to-end performance. Consequently, workflow outcomes are often strongly influenced by bottleneck stages whose failure probabilities dominate overall execution risk. Although conditional independence provides a tractable baseline, the framework can be extended to incorporate explicit error-propagation mechanisms, correlated failures, and inter-stage dependencies, as in cross-layer models of secure and resilient cyber-physical systems~\cite{ZhuBasar2015,ZhuXu2020CrossLayer,Zhu2024CyberResilience}.

\subsection{Workflow Value}

The economic value generated by a workflow is generally a property of the entire team and its induced configuration, rather than a sum of isolated stage-level values. We define the gross workflow-value function as

\begin{equation}
\label{eq:workflow-value}
V:\mathcal X\rightarrow\mathbb R,
\end{equation}

\noindent where \(V(x)\) represents the expected value generated by configuration \(x\) before subtracting deployment costs and failure losses. This quantity captures the productive contribution of the team participating in the workflow. Depending on the application domain, workflow value may represent revenue generation, productivity gains, mission effectiveness, quality improvements, risk reduction, or other organizational objectives. The function \(V\) may encode complementarities, substitution effects, bottlenecks, network effects, nonlinear production, and other interactions among workflow stages.

An additive stage-value model is a useful special case. If \(v_k(x_k)\) denotes the expected value generated locally at stage \(k\), then one may set \(V(x)=\sum_{k=1}^{K}v_k(x_k)\). The framework does not require this separability assumption.

\subsection{Workflow Cost}

The total economic cost of a workflow configuration is obtained by aggregating stage-level deployment costs. Specifically,

\begin{equation}
\label{eq:workflow-cost}
C(x)
=
\sum_{k=1}^{K} c_k(x_k).
\end{equation}

\noindent This quantity includes all expenditures associated with acquiring and deploying agents, including labor costs, licensing fees, computational expenses, infrastructure utilization, monitoring requirements, and supervisory overhead. In many applications, workflow failure generates additional economic consequences beyond direct execution costs. Let \(L_F>0\) denote the loss incurred when the workflow fails. This loss may represent financial damages, regulatory penalties, safety consequences, reputational harm, or operational disruptions. The expected failure loss is therefore

\begin{equation}
\label{eq:workflow-failure-loss}
L(x)
=
L_F\,q(x)
=
L_F\left(1-R(x)\right).
\end{equation}

\noindent Substituting the reliability expression yields

\begin{equation}
\label{eq:workflow-failure-loss-expanded}
L(x)
=
L_F
\left(
1-
\prod_{k=1}^{K}
\rho_k(x_k)
\right).
\end{equation}

\subsection{Net Workflow Value}

The central object of Agentomics is the net economic value generated by a workflow after accounting for deployment costs and execution risk. We therefore define the net workflow value as

\begin{equation}
\label{eq:workflow-net-value}
W(x)
=
V(x)
-
C(x)
-
L(x).
\end{equation}

\noindent Substituting the preceding definitions yields

\begin{equation}
\label{eq:workflow-net-value-expanded}
W(x)
=
V(x)
-
\sum_{k=1}^{K}
c_k(x_k)
-
L_F
\left(
1-
\prod_{k=1}^{K}
\rho_k(x_k)
\right).
\end{equation}

\noindent The quantity \(W(x)\) represents the expected economic surplus generated by workflow configuration \(x\). Configurations producing greater team-level value, lower deployment costs, and higher reliability achieve larger values of \(W(x)\). If the special additive model \(V(x)=\sum_{k=1}^{K}v_k(x_k)\) is adopted, then \eqref{eq:workflow-net-value-expanded} reduces to the corresponding stage-sum expression.

\subsection{Economic Interpretation}

The workflow designer faces a resource-allocation problem in which human and AI agents must be assigned to stages, each with its own task \(\tau_k\), so as to maximize organizational value. The economic attractiveness of AI deployment therefore depends not only on whether an AI agent can perform the task associated with a stage, but also on whether the resulting increase in value exceeds any associated deployment costs and reliability losses.

In high-consequence domains such as healthcare, cybersecurity, critical infrastructure, and autonomous systems, large failure losses may justify extensive human oversight despite higher labor costs. Conversely, in environments characterized by low failure penalties and large-scale repetitive workloads, AI-based automation may generate substantial economic gains.

The Agentomics perspective therefore shifts attention from component-level benchmarking to system-level value creation. The relevant question is not whether a particular AI agent outperforms a human agent on an isolated task, but whether a given allocation of human and artificial agents maximizes the net economic value generated by the workflow as a whole.

\section{Human--AI Coalitions and Hybrid Workflows}
\label{sec:coalitions}

The preceding sections introduced workflow configurations together with the economic and performance characteristics of individual agents. We now extend the framework to study the incremental deployment of AI agents into an existing human-operated workflow. Rather than viewing automation as an all-or-nothing decision, Agentomics models deployment as a coalition-formation process in which different subsets of AI agents participate in workflow execution.

This perspective is important because the economic contribution of an AI agent is inherently system dependent. In sequential workflows, the value generated by an agent depends not only on its own capabilities but also on its interactions with upstream information, downstream decision-making processes, human supervision, and the reliability characteristics of other agents in the workflow. Consequently, an AI system that appears attractive when evaluated in isolation may fail to improve overall workflow performance once coordination effects, reliability losses, and economic risk are taken into account.

The objective of this section is to formalize hybrid human--AI workflows as coalition-induced configurations and to establish a rigorous notion of agent value. This formulation naturally leads to a cooperative-game representation in which the economic gains generated by automation can be attributed among participating AI agents. The resulting attribution provides the foundation for agent valuation and pricing within the Agentomics framework.

\subsection{Coalitional Structure}

Consider a workflow consisting of \(K\) sequential tasks \(\tau_1 \rightarrow \tau_2 \rightarrow \cdots \rightarrow \tau_K\). 
The benchmark human workflow configuration is denoted by \(x^H=(H_1,\ldots,H_K)\), where \(H_k\in\mathcal H\) is the human operator assigned to stage \(k\). This benchmark represents the incumbent workflow against which all AI-enabled configurations are evaluated.

Suppose that the workflow operator has access to a collection of deployable AI agents \(N=\{A_1,\ldots,A_n\}\subseteq\mathcal I\). A coalition \(S\subseteq N\) represents the subset of AI agents activated within the workflow. The empty coalition \(S=\emptyset\) corresponds to the benchmark human workflow, while the grand coalition \(S=N\) corresponds to maximal AI deployment. Intermediate coalitions induce hybrid human--AI workflows and therefore capture the gradual adoption of AI systems within operational environments.

The coalition perspective recognizes that workflow performance is fundamentally a system-level property. The economic value generated by an AI agent depends on how its outputs interact with the remainder of the workflow, implying that agent contributions are generally coalition dependent.

\subsection{Coalition-Induced Workflow Configurations}

Each coalition \(S\subseteq N\) induces a workflow configuration \(x^S=(x_1^S,\ldots,x_K^S)\), where \(x_k^S\in\mathcal A_k\) denotes the agent assigned to stage \(k\). 
The deployment mapping \(\Gamma:2^N \rightarrow \mathcal X\) satisfies \(x^S=\Gamma(S)\). 
The mapping \(\Gamma\) specifies how activated AI agents are allocated across workflow stages. It encodes the feasibility constraints of the deployment problem, including cases in which multiple AI agents are feasible candidates for the same stage and the operator must select among them. If no AI agent in \(S\) is assigned to stage \(k\), then the benchmark human operator \(H_k\) remains responsible for executing \(\tau_k\). Consequently, each coalition induces a distinct workflow architecture with its own cost structure, reliability profile, information flow, and operational characteristics.

The benchmark workflow is recovered as \(x^\emptyset=x^H\), while \(x^N\) denotes the workflow obtained under full AI deployment. These two configurations are useful reference points: \(x^\emptyset\) anchors the coalition game at the incumbent human workflow, whereas \(x^N\) represents the maximal deployment scenario against which total AI-enabled surplus is measured.

\subsection{Workflow Economics Under Coalition Formation}

Given a coalition \(S\subseteq N\), its induced configuration \(x^S\) is evaluated using the workflow-level quantities in Section~\ref{sec:workflow-economics}. In particular, \(V(x^S)\) denotes the gross value generated by the entire team under configuration \(x^S\), the reliability is obtained from \eqref{eq:workflow-reliability} as \(R(x^S)=\prod_{k=1}^{K}\rho_k(x_k^S)\), the failure probability is \(q(x^S)=1-R(x^S)\), and the net workflow value is \(W(x^S)=V(x^S)-C(x^S)-L(x^S)\). Thus \(W(x^S)\) represents the expected economic surplus generated by coalition \(S\). No additive decomposition of \(V(x^S)\) across stages is assumed.

\subsection{Coalition Value}

The economic value generated by AI deployment is measured relative to the benchmark human workflow. This normalization ensures that coalition values represent incremental surplus rather than absolute workflow output, which is important when comparing agents across workflows with different baseline productivity levels.

\begin{definition}[Coalition Value]
The coalition-value function \(g:2^N\rightarrow\mathbb R\) is defined by

\begin{equation}
\label{eq:coalition-value}
g(S)
=
W(x^S)-W(x^H).
\end{equation}
\end{definition}

\noindent The quantity \(g(S)\) measures the incremental economic surplus generated by coalition \(S\) relative to the benchmark human workflow. 
A coalition is economically beneficial whenever \(g(S)\ge 0\). The set of economically admissible coalitions is 
\(
\mathcal C
=
\left\{
S\subseteq N:
g(S)\ge 0
\right\}.
\)

\begin{proposition}
A coalition \(S\) is economically admissible if and only if \(W(x^S)\ge W(x^H)\).
\end{proposition}

\begin{proof}
The result follows immediately from the definition \(g(S)=W(x^S)-W(x^H)\).
\end{proof}

\subsection{Agent Economic Value}

The deployment problem induces the transferable-utility cooperative game \((N,g)\). 
For any coalition \(S\subseteq N\setminus\{i\}\), the marginal contribution of agent \(i\) is \(g(S\cup\{i\})-g(S)\). Since workflow performance depends on interactions among agents, marginal contributions generally depend on coalition composition.

\begin{definition}[Agent Economic Value]
For each AI agent \(A_i\in N\), the economic value of the agent is its Shapley value \(\phi_i(g)\), where
\begin{equation}
\label{eq:shapley-value}
\phi_i(g)
=
\sum_{S\subseteq N\setminus\{i\}}
\frac{|S|!(|N|-|S|-1)!}{|N|!}
\left[
g(S\cup\{i\})-g(S)
\right].
\end{equation}
\end{definition}

\noindent The quantity \(\phi_i(g)\) represents the expected marginal economic surplus generated by agent \(A_i\) across all possible coalition-formation orders. It provides a system-level measure of agent productivity that accounts for workflow interactions, complementarities, and substitution effects.

The Shapley values satisfy 
\(
\sum_{A_i\in N} \phi_i(g)=g(N),
\) 
which implies that the total economic surplus generated by the fully deployed workflow is completely allocated among participating agents.

\begin{proposition}[Sign Interpretation of Agent Value]
The economic value of agent \(A_i\) is denoted by \(\phi_i(g)\).

\begin{enumerate}
\item If \(\phi_i(g)>0\), then agent \(A_i\) increases expected workflow surplus.
\item If \(\phi_i(g)=0\), then agent \(A_i\) is economically neutral.
\item If \(\phi_i(g)<0\), then agent \(A_i\) decreases expected workflow surplus.
\end{enumerate}
\end{proposition}

\begin{proof}
The Shapley value equals the expected marginal contribution of an agent across all coalition-formation orders. Positive values imply positive expected contributions to workflow surplus, zero values imply no expected contribution, and negative values imply negative expected contributions.
\end{proof}

A negative economic value identifies a value-destroying agent. Such an agent may introduce reliability degradation, coordination costs, operational inefficiencies, or other adverse effects that outweigh any local performance gains. Consequently, the deployment of such an agent is not economically justified unless its behavior or operating environment changes.

\subsection{Agent Pricing}

The Agentomics framework distinguishes between the economic value generated by an agent and the market price charged for its deployment. Let \(p_i\ge 0\) denote the market price of agent \(A_i\). The pair \((\phi_i(g),p_i)\) characterizes the economic position of the agent within the marketplace. An agent is economically attractive whenever \(\phi_i(g)>p_i\), meaning that the agent generates more economic surplus than it costs to deploy. Conversely, if \(\phi_i(g)<p_i\), then the deployment cost exceeds the expected value generated by the agent.

\begin{definition}[Shapley Pricing Equilibrium]
An agent market is said to be in Shapley pricing equilibrium if
\begin{equation}
\label{eq:shapley-pricing-equilibrium}
p_i=\phi_i(g)
\end{equation}
for every agent \(A_i\in N\).
\end{definition}

\noindent Under Shapley pricing equilibrium, the price of an agent equals its expected marginal contribution to workflow surplus. This equilibrium should be interpreted as a normative benchmark rather than a prediction that all observed prices will immediately satisfy the condition.

\begin{proposition}[Marginal Productivity Pricing]
In a Shapley pricing equilibrium, every agent is compensated according to its expected contribution to workflow value.
\end{proposition}

\begin{proof}
By definition, \(p_i=\phi_i(g)\). Since the Shapley value equals the expected marginal contribution of agent \(A_i\), the agent's price equals its expected economic contribution.
\end{proof}

More generally, market frictions, bargaining power, informational asymmetries, competition among providers, and strategic pricing behavior may cause prices to deviate from Shapley-based values. The quantity \(p_i-\phi_i(g)\) may therefore be interpreted as a measure of overvaluation or undervaluation relative to the agent's economic contribution.

The Agentomics framework thus establishes a direct connection between workflow economics, cooperative-game-theoretic value attribution, and market pricing. Rather than evaluating AI systems solely through benchmark performance, agents are assessed according to the economic surplus they contribute to end-to-end workflow outcomes and the prices they command in competitive markets.

\section{Case Study: Economic Valuation and Pricing of AI Agents in a Security Operations Center}
\label{sec:case-study}

To illustrate the Agentomics framework and the coalition-based valuation theory developed in Section~\ref{sec:coalitions}, consider a medium-sized enterprise operating a Security Operations Center (SOC) responsible for monitoring network activity and responding to cybersecurity incidents. The objective of the SOC is to identify malicious activity, investigate potential threats, and recommend appropriate response actions while minimizing operational costs and cyber risk. This setting is especially natural for agentic AI because cyber defense increasingly involves strategic interactions among defenders, attackers, automated tools, and human analysts~\cite{Zhu2025LLMAgenticCybersecurity,AlBariZhu2025Gestalt}.

This example demonstrates how workflow economics, coalition formation, Shapley-value attribution, accountability, and pricing interact within a realistic operational environment. More importantly, it illustrates the central claim of Section~\ref{sec:coalitions}: the value of an AI agent depends on its contribution to workflow-level economic outcomes rather than on isolated measures of technical performance.

The incident-handling workflow consists of four sequential stages indexed by \(k=1,\ldots,4\). Stage \(k\) has task \(\tau_k\): \(\tau_1\) is alert detection, \(\tau_2\) alert triage, \(\tau_3\) threat investigation, and \(\tau_4\) incident response recommendation. Historically, each stage is performed by a human analyst, yielding the benchmark workflow \(x^H=(H_1,H_2,H_3,H_4)\). The numerical values in this case study are calibrated illustrative assumptions designed to demonstrate the mechanics of the framework rather than estimates from a particular SOC deployment. Assume that the SOC processes approximately 1,000 alerts per day.

\begin{figure*}[t]
\centering
\includegraphics[width=0.95\textwidth]{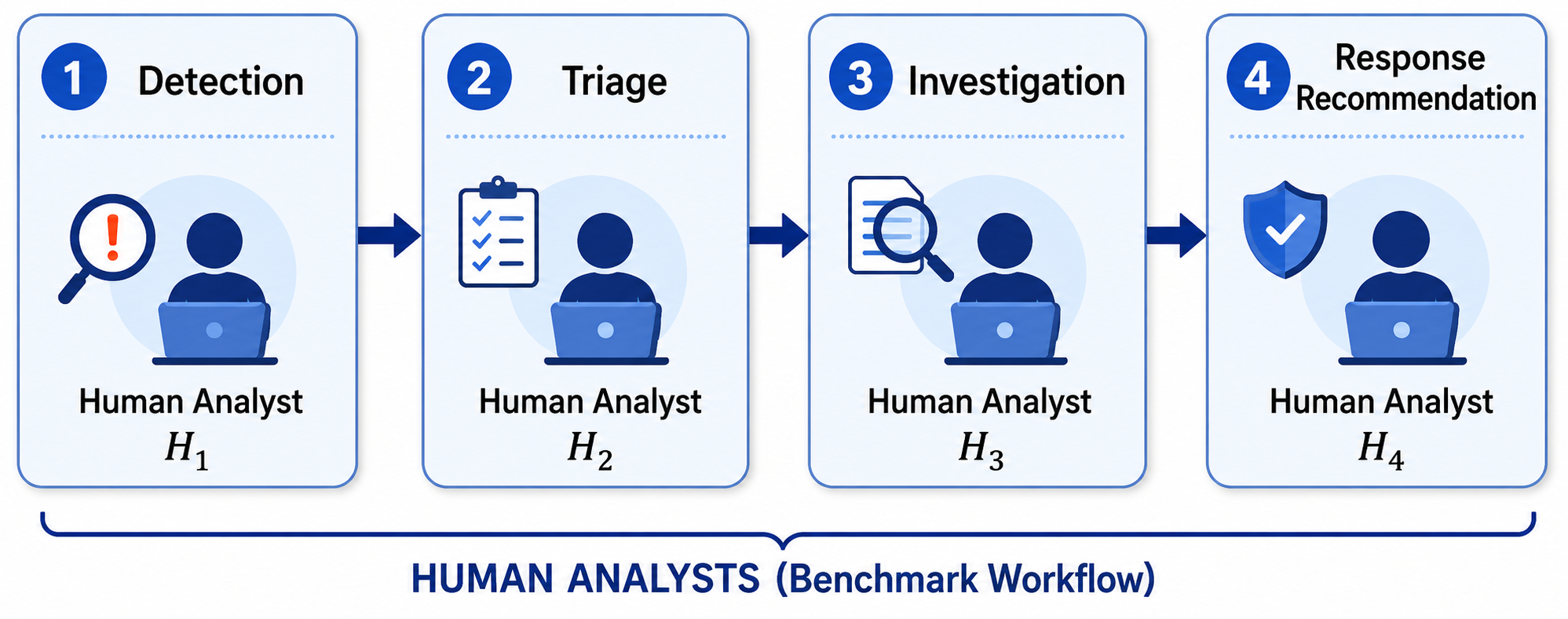}
\caption{
Benchmark human-operated SOC workflow.
Each stage of the incident-handling process is performed by a
human analyst. The benchmark workflow
\(x^H=(H_1,H_2,H_3,H_4)\) serves as the reference configuration
against which all AI-enabled coalition structures are evaluated.
}
\label{fig:soc-human-workflow}
\end{figure*}
Figure~\ref{fig:soc-human-workflow} illustrates the benchmark
workflow in which all four stages are performed by human
analysts. This configuration corresponds to the empty coalition
\(S=\emptyset\) and establishes the baseline economic value
against which alternative human--AI organizational structures
are compared.

\begin{figure*}[t]
\centering
\includegraphics[width=0.95\textwidth]{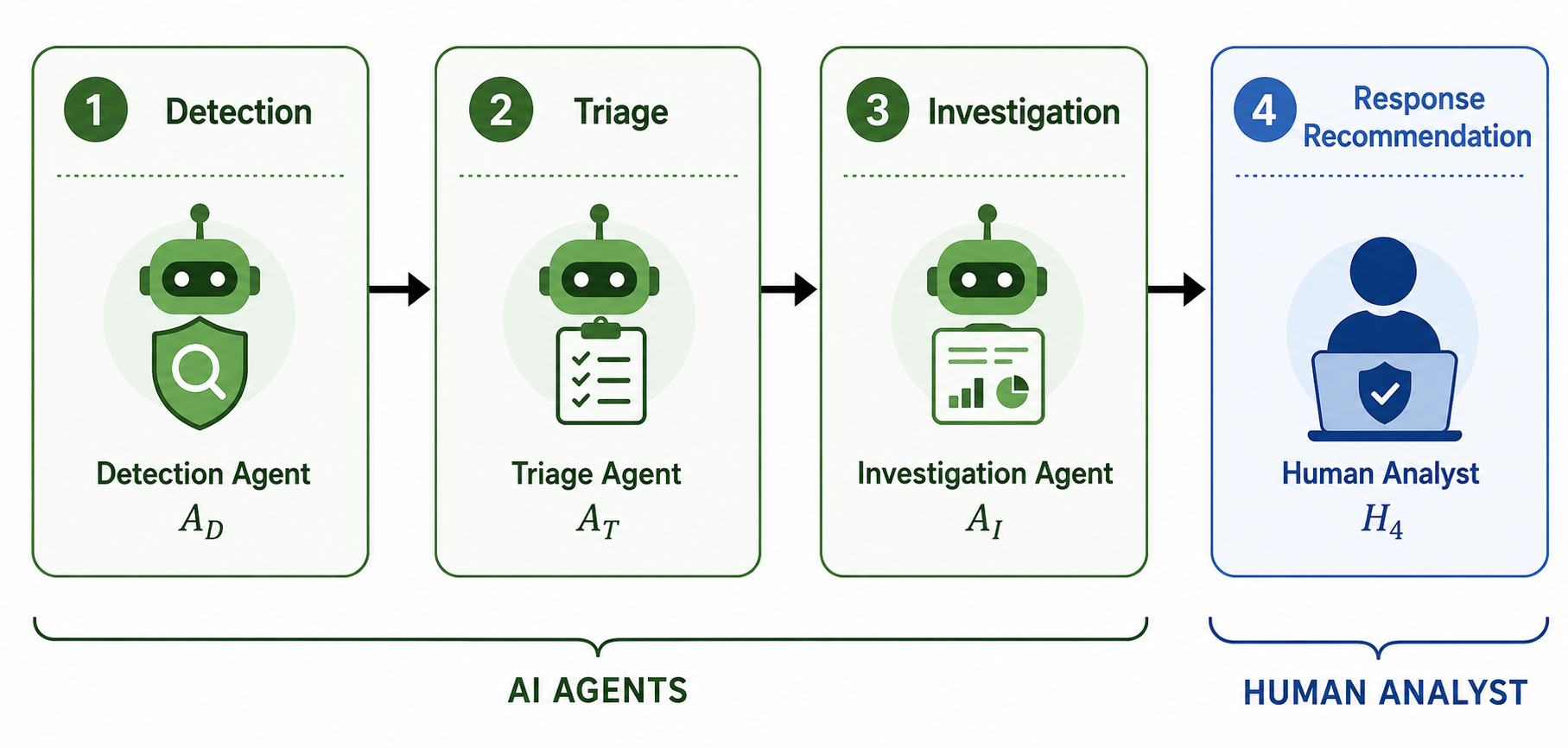}
\caption{
Hybrid human--AI workflow corresponding to the grand coalition.
The Detection Agent (\(A_D\)), Triage Agent (\(A_T\)), and
Investigation Agent (\(A_I\)) automate the first three stages
of the workflow, while a human analyst (\(H_4\)) retains
responsibility for the final response recommendation.
This configuration illustrates how AI agents and humans
cooperate to create workflow value.
}
\label{fig:soc-hybrid-workflow}
\end{figure*}

Figure~\ref{fig:soc-hybrid-workflow} illustrates the grand coalition workflow considered in the case study. The first three stages are automated by AI agents while the final decision remains under human supervision. This structure reflects a common deployment pattern observed in operational cybersecurity environments, where AI systems perform information processing and prioritization tasks while human analysts retain responsibility for high-consequence decisions. From the perspective of Agentomics, the transition from Figure~\ref{fig:soc-human-workflow} to Figure~\ref{fig:soc-hybrid-workflow} represents a coalition formation process. The resulting economic surplus generated by automation is measured through the coalition-value function and allocated among participating agents using the Shapley-value framework developed in Section~\ref{sec:coalitions}.

The benchmark workflow shown in Figure~\ref{fig:soc-human-workflow} produces a net workflow value of \(W(x^H)=\$2{,}845\). The hybrid workflow shown in Figure~\ref{fig:soc-hybrid-workflow} produces \(W(x^{\{A_D,A_T,A_I\}})=\$4{,}560\). By the coalition-value definition in \eqref{eq:coalition-value}, the grand coalition generates \(g(\{A_D,A_T,A_I\})=4{,}560-2{,}845=\$1{,}715\). This surplus represents the additional economic value created by the coalition of AI agents relative to the benchmark human-operated workflow.

\subsection{Human-Only Workflow Economics}

As an external salary anchor, the U.S. Bureau of Labor Statistics reports median annual pay for information security analysts of roughly \$125,000 in 2024~\cite{BLSInfoSec2025}. We use a rounded salary benchmark of \$120,000 for a Tier-2 SOC analyst and assume a 1.5 overhead multiplier to account for benefits, training, management, equipment, office space, and infrastructure support. This yields a fully burdened annual cost of approximately \$180,000 per analyst.

Assuming 250 working days per year, the effective labor cost is \(c_H=180{,}000/250=\$720\) per analyst per day. Since the benchmark workflow requires four analysts, the total daily labor cost is \(C(x^H)=4\times720=\$2{,}880\). Suppose that successful processing of an alert produces approximately \$15 of organizational value through reduced analyst effort, improved security posture, faster incident resolution, and reduced business disruption. Processing 1,000 alerts therefore generates a daily gross workflow value of \(V(x^H)=1000\times15=\$15{,}000\). Human analysts are highly reliable but not perfect. If each analyst successfully completes the assigned stage with probability \(\rho_H=0.95\), then \eqref{eq:workflow-reliability} gives \(R(x^H)=0.95^4=0.8145\). Thus approximately 18.5\% of workflows experience some form of escalation, delay, misclassification, or investigation error.

Suppose that workflow failures generate an average expected loss of \$50,000 through missed attacks, delayed responses, operational disruption, compliance violations, and business interruption. The expected daily failure loss is \(L(x^H)=50{,}000(1-0.8145)=\$9{,}275\). Following the net-value definition in \eqref{eq:workflow-net-value}, substituting the benchmark values yields \(W(x^H)=15{,}000-2{,}880-9{,}275=\$2{,}845\).

This benchmark serves as the reference point against which all AI-enabled workflows will be evaluated.

\subsection{Available AI Agents}

Suppose the organization evaluates three commercially available AI agents, \(N=\{A_D,A_T,A_I\}\), where \(A_D\), \(A_T\), and \(A_I\) denote the Detection, Triage, and Investigation Agents, respectively. In this case study, each of the first three stages has one designated AI candidate: \(A_D\) is feasible for stage 1 with task \(\tau_1\), \(A_T\) is feasible for stage 2 with task \(\tau_2\), and \(A_I\) is feasible for stage 3 with task \(\tau_3\). The fourth stage, whose task is \(\tau_4\), remains human-operated. This one-AI-candidate-per-stage structure is a simplifying assumption for the case study, not a restriction of the general framework, which allows multiple AI candidates in the same feasible set \(\mathcal A_k\). The Detection Agent automatically identifies suspicious alerts from network telemetry and log data. The Triage Agent prioritizes alerts according to estimated threat severity and business impact. The Investigation Agent gathers contextual information, correlates indicators of compromise, and generates investigation summaries for analysts.

Assume annual enterprise subscription costs of \$18,000, \$12,000, and \$24,000 for the Detection, Triage, and Investigation Agents, respectively. These correspond to daily subscription costs of approximately \$72, \$48, and \$96. In addition to subscription fees, the agents incur inference costs associated with LLM usage. Suppose the Detection Agent consumes 20 million tokens per day, the Triage Agent 10 million tokens per day, and the Investigation Agent 40 million tokens per day. Using a representative blended inference price of \$5 per million tokens, the resulting daily inference costs are \$100, \$50, and \$200. These subscription, usage, and reliability values are not intended to represent a specific vendor quote; they are scenario parameters chosen to illustrate how changes in cost and reliability propagate through the Agentomics valuation model. Table~\ref{tab:soc-agent-assumptions} summarizes the resulting daily deployment costs and reliability assumptions.

\begin{table}[t]
\centering
\caption{Agent-level cost and reliability assumptions for the SOC case study.}
\label{tab:soc-agent-assumptions}
\begin{tabular}{lcc}
\toprule
Agent & Daily deployment cost & Stage reliability \\
\midrule
Human analyst & \$720 & 0.95 \\
Detection Agent \(A_D\) & \$172 & 0.98 \\
Triage Agent \(A_T\) & \$98 & 0.93 \\
Investigation Agent \(A_I\) & \$296 & 0.90 \\
\bottomrule
\end{tabular}
\end{table}

Although the Detection Agent exceeds human performance, the Triage and Investigation Agents exhibit lower reliability. Consequently, introducing AI may reduce costs while simultaneously increasing failure risk. Determining whether deployment is economically justified therefore requires evaluating productivity gains, deployment costs, reliability effects, and coalition-level complementarities together.

This observation directly reflects the net-value expression in \eqref{eq:workflow-net-value-expanded}. An AI agent may reduce labor costs and improve workflow speed while simultaneously increasing expected failure losses. Consequently, technical capability and economic value are not equivalent concepts.

\subsection{Why Coalition Formation Matters}

A central idea of Section~\ref{sec:coalitions} is that the value of an AI agent cannot generally be determined in isolation. In a workflow, agents interact with upstream information sources, downstream decision makers, human supervisors, and other AI systems. Consequently, the economic contribution of an agent depends on the organizational environment in which it operates.

Consider the Detection Agent. If analysts are unable to process generated alerts efficiently, the agent may create little economic value despite achieving excellent technical performance. Conversely, when paired with an effective triage process and experienced analysts, the same agent may substantially improve workflow performance.

This dependence on organizational context motivates the coalition-based perspective developed in Section~\ref{sec:coalitions}. Rather than asking whether an individual agent performs well in isolation, Agentomics asks how much incremental economic value is created when an agent participates in a particular coalition of humans and AI systems.

From this perspective, each coalition represents a distinct organizational design. Different coalitions correspond to different allocations of responsibilities between humans and machines and therefore generate different levels of value, cost, reliability, and risk.

\subsection{Coalition-Based Deployment}

Table~\ref{tab:soc-coalition-values} reports the coalition-level calculations for all \(2^3\) AI coalitions. Reliability is computed from \eqref{eq:workflow-reliability}; deployment cost aggregates the human and AI costs in Table~\ref{tab:soc-agent-assumptions}; failure loss is \(L(x^S)=50{,}000(1-R(x^S))\); and net value is computed as \(W(x^S)=V(x^S)-C(x^S)-L(x^S)\). Because \(V(x^S)\) is a team-level workflow value, not a sum of isolated stage values, the table allows gross value to vary by coalition in order to reflect productivity, speed, information quality, and complementarity effects. Reliability values are rounded to four decimal places, while monetary values are computed from unrounded reliability products and rounded to the nearest dollar.

\begin{table}[t]
\centering
\caption{Coalition-level workflow economics for the SOC case study.}
\label{tab:soc-coalition-values}
\resizebox{\textwidth}{!}{%
\begin{tabular}{llrrrrrr}
\toprule
Coalition \(S\) & Configuration \(x^S\) & \(R(x^S)\) & \(C(x^S)\) & \(L(x^S)\) & \(V(x^S)\) & \(W(x^S)\) & \(g(S)\) \\
\midrule
\(\emptyset\) & \((H_1,H_2,H_3,H_4)\) & 0.8145 & \$2,880 & \$9,275 & \$15,000 & \$2,845 & \$0 \\
\(\{A_D\}\) & \((A_D,H_2,H_3,H_4)\) & 0.8402 & \$2,332 & \$7,989 & \$13,771 & \$3,450 & \$605 \\
\(\{A_T\}\) & \((H_1,A_T,H_3,H_4)\) & 0.7974 & \$2,258 & \$10,132 & \$15,485 & \$3,095 & \$250 \\
\(\{A_I\}\) & \((H_1,H_2,A_I,H_4)\) & 0.7716 & \$2,456 & \$11,418 & \$16,819 & \$2,945 & \$100 \\
\(\{A_D,A_T\}\) & \((A_D,A_T,H_3,H_4)\) & 0.8225 & \$1,710 & \$8,873 & \$14,703 & \$4,120 & \$1,275 \\
\(\{A_D,A_I\}\) & \((A_D,H_2,A_I,H_4)\) & 0.7960 & \$1,908 & \$10,200 & \$15,708 & \$3,600 & \$755 \\
\(\{A_T,A_I\}\) & \((H_1,A_T,A_I,H_4)\) & 0.7554 & \$1,834 & \$12,230 & \$17,639 & \$3,575 & \$730 \\
\(\{A_D,A_T,A_I\}\) & \((A_D,A_T,A_I,H_4)\) & 0.7792 & \$1,286 & \$11,038 & \$16,884 & \$4,560 & \$1,715 \\
\bottomrule
\end{tabular}}
\end{table}

The benchmark workflow produces a net value of \$2,845 per day. Deploying the Detection Agent alone increases workflow value to \$3,450, generating an incremental surplus of \$605. Deploying the grand coalition increases net value to \$4,560 and generates an incremental surplus of \$1,715. These values quantify the economic surplus generated by AI deployment relative to the benchmark human workflow.

The economically preferred configuration in this case study is therefore the grand coalition \(S=\{A_D,A_T,A_I\}\), corresponding to \(x^S=(A_D,A_T,A_I,H_4)\). It maximizes net workflow value among the configurations in Table~\ref{tab:soc-coalition-values}. This conclusion would not be obtained by maximizing reliability alone: the Detection-only configuration has the highest reliability, \(R(x^{\{A_D\}})=0.8402\), whereas the grand coalition has \(R(x^{\{A_D,A_T,A_I\}})=0.7792\). The case study therefore illustrates that the best human--AI configuration is determined by the full economic objective \(W(x)\), which jointly accounts for gross workflow value, operating cost, and expected failure loss.

\subsection{Economic Interpretation of Coalition Value}

The coalition values computed above represent the central economic quantity in Agentomics. A key insight is that valuation should be performed using incremental surplus rather than technical performance metrics. An agent may achieve impressive benchmark scores while generating little economic surplus once costs and reliability effects are incorporated. Conversely, an agent with modest standalone performance may generate substantial economic value if it improves workflow efficiency or complements other participants. Thus, the coalition-value function in \eqref{eq:coalition-value} plays a role analogous to a profit function in production economics: it measures the additional economic surplus generated by a particular organizational design after accounting for the interaction between productivity, cost, and risk.

\subsection{Marginal Contribution and Shapley-Based Economic Attribution}

The central question is how the total surplus of \$1,715 should be attributed among the participating agents. This attribution problem is fundamentally challenging because workflow outcomes emerge collectively. Although the grand coalition generates \$1,715 of surplus, it is not immediately obvious how much of that value should be assigned to the Detection Agent, the Triage Agent, or the Investigation Agent. One possibility would be to allocate value according to standalone performance, but such an approach ignores complementarities and interaction effects. The Detection Agent may create substantially more value when paired with the Triage Agent than when operating alone, while the Investigation Agent derives much of its usefulness from information generated by upstream stages.

To address this challenge, Section~\ref{sec:coalitions} introduced the Shapley value in \eqref{eq:shapley-value}. The Shapley value may be interpreted as the expected marginal productivity of an agent across all possible organizational structures. Rather than measuring isolated performance, it measures expected contribution to collective value creation by averaging an agent's incremental contribution over all possible orders in which agents could be added to the workflow.

Applying \eqref{eq:shapley-value} to the coalition values in Table~\ref{tab:soc-coalition-values} yields the attribution results in Table~\ref{tab:soc-shapley-pricing}. These values satisfy the efficiency property because \(810+620+285=1{,}715\), which exactly equals the total grand-coalition surplus. The economic interpretation is straightforward: the Detection Agent contributes approximately 47.2\% of total workflow surplus, the Triage Agent contributes approximately 36.2\%, and the Investigation Agent contributes approximately 16.6\%. Importantly, these contributions are not determined solely by individual performance. They reflect the marginal value created by each agent across all possible deployment configurations, allowing the Shapley value to capture complementarities, substitution effects, and workflow interactions that are invisible to conventional benchmark evaluations.

\subsection{Accountability and Responsibility Attribution}

Consider a ransomware incident that successfully bypasses organizational defenses because a malicious alert was incorrectly prioritized during triage. Traditional performance metrics may identify the stage at which the error occurred, but they provide limited guidance for assigning responsibility across the larger workflow. The Agentomics framework separates two related concepts. The first is causal accountability: if audit logs show that the triage decision was the proximate error, then the Triage Agent \(A_T\), together with the organizational actor responsible for deploying and supervising it, is the key accountable player for the incident. The second is systemic accountability: because workflow risk is jointly produced by interacting agents, part of the realized loss can be allocated across the coalition according to each agent's contribution to workflow outcomes.

To make this distinction concrete, suppose the realized incident loss is \(F=\$50{,}000\). Using the Shapley shares from Table~\ref{tab:soc-shapley-pricing} as a systemic accountability baseline assigns 47.2\% of the loss to the Detection Agent, 36.2\% to the Triage Agent, and 16.6\% to the Investigation Agent. Table~\ref{tab:soc-accountability-fault} reports the resulting accountability allocation and translates it into a daily risk adjustment by amortizing the incident over a 100-day operating horizon. Under this systemic allocation, the Triage Agent receives an accountability debit of approximately \$181 per day, reducing its risk-adjusted economic contribution from \$620 to \$439 per day. Under a stricter localized causal rule, the entire incident would be charged to \(A_T\), yielding a \$500 daily debit over the same horizon and a risk-adjusted contribution of only \$120 per day. The choice between these rules is a governance decision: causal attribution identifies the primary agent to audit or replace, while systemic attribution quantifies how the loss should be economically shared across the deployed AI coalition.

\begin{table}[t]
\centering
\caption{Illustrative accountability allocation for a triage fault causing a \$50,000 incident loss.}
\label{tab:soc-accountability-fault}
\resizebox{\textwidth}{!}{%
\begin{tabular}{lrrrrr}
\toprule
Agent & Shapley share & Baseline value & Loss allocation & 100-day debit & Risk-adjusted value \\
\midrule
Detection Agent \(A_D\) & 47.2\% & \$810 & \$23,615 & \$236 & \$574 \\
Triage Agent \(A_T\) & 36.2\% & \$620 & \$18,076 & \$181 & \$439 \\
Investigation Agent \(A_I\) & 16.6\% & \$285 & \$8,309 & \$83 & \$202 \\
\bottomrule
\end{tabular}}
\end{table}

This observation reflects a broader principle developed throughout the paper: accountability requires attribution. An organization cannot determine which agents should be rewarded, trusted, regulated, replaced, or audited without first understanding how those agents contribute to workflow outcomes. The ability to decompose workflow outcomes into agent-level contributions therefore provides a principled foundation for governance, auditing, accountability, and regulatory oversight in hybrid human--AI systems.

\subsection{From Marginal Productivity to Agent Pricing}

\begin{table}[t]
\centering
\caption{Shapley attribution and pricing comparison.}
\label{tab:soc-shapley-pricing}
\resizebox{\textwidth}{!}{%
\begin{tabular}{lrrrrl}
\toprule
Agent & Shapley value & Surplus share & Market price & \(\phi_i(g)-p_i\) & Pricing implication \\
\midrule
Detection Agent \(A_D\) & \$810 & 47.2\% & \$500 & \$310 & Undervalued \\
Triage Agent \(A_T\) & \$620 & 36.2\% & \$900 & \(-\$280\) & Overvalued \\
Investigation Agent \(A_I\) & \$285 & 16.6\% & \$250 & \$35 & Near fair value \\
\bottomrule
\end{tabular}}
\end{table}

Suppose that the marketplace prices for the three agents are \(p_D=\$500\), \(p_T=\$900\), and \(p_I=\$250\) per day. Comparing prices with economic contributions yields the valuation gaps shown in Table~\ref{tab:soc-shapley-pricing}. The Detection Agent appears undervalued because it generates substantially more economic value than its market price suggests. The Triage Agent appears overvalued because its price exceeds its economic contribution. The Investigation Agent is approximately fairly priced.

The accountability calculation above also affects pricing. Before the incident, the Triage Agent is already overvalued by \$280 per day because its market price of \$900 exceeds its Shapley value of \$620. After the triage fault, the systemic accountability adjustment in Table~\ref{tab:soc-accountability-fault} lowers its risk-adjusted value to \$439, increasing the pricing gap to \$461 per day. Under the stricter localized causal rule, its risk-adjusted value would fall to \$120, increasing the pricing gap to \$780 per day. Thus the same attribution framework identifies \(A_T\) as the key accountable agent for the fault and shows how accountability should be reflected in contract pricing, insurance premia, service-level credits, or future procurement decisions.

Classical labor economics often interprets wages through the lens of marginal productivity: workers who contribute more to production generally command higher compensation. Agentomics extends this principle to AI systems by tying compensation to economic contribution rather than to benchmark scores, model size, token consumption, or vendor reputation. Under the proposed Shapley Pricing Equilibrium in \eqref{eq:shapley-pricing-equilibrium}, equilibrium prices are \(p_D^*=810\), \(p_T^*=620\), and \(p_I^*=285\). These prices are determined by economic contribution rather than by subscription heuristics, and the equilibrium may therefore be viewed as a generalization of marginal productivity pricing from traditional labor markets to agentic AI ecosystems.

\subsection{Lessons Learned}

This case study illustrates several central principles of Agentomics. First, technical performance alone is insufficient for determining economic value. Second, the economically preferred configuration is the one that maximizes net workflow value rather than reliability or cost in isolation; in the SOC example, this is the grand coalition \((A_D,A_T,A_I,H_4)\). Third, economic value emerges from interactions among agents rather than from isolated capabilities. Fourth, cooperative-game-theoretic attribution provides a rigorous mechanism for allocating both value and responsibility. Fifth, accountability and valuation are fundamentally linked because both require attribution of contribution. Finally, agent pricing can be grounded in measurable economic contribution rather than heuristic subscription models and can be adjusted when faults reveal realized accountability exposure. More broadly, the example demonstrates how Agentomics provides a unified framework connecting workflow economics, coalition formation, reliability analysis, accountability, value attribution, and market pricing in hybrid human--AI systems.

\section{Conclusions}
\label{sec:conclusion}

This paper introduced Agentomics as a workflow-centered framework for valuing, attributing, and pricing human and artificial agents. The central premise is that the economic value of an AI agent is not an intrinsic property of the agent alone, nor can it be inferred directly from benchmark performance. Rather, agent value depends on the contribution the agent makes to end-to-end workflow outcomes when embedded within a particular organizational configuration. By modeling gross workflow value, deployment costs, workflow reliability, and expected failure losses, the framework provides a unified measure of net workflow surplus for hybrid human--AI systems.

The framework emphasizes that workflow value is a team-level object. In general, the value generated by a workflow need not be separable across stages because complementarities, substitution effects, bottlenecks, supervision, information quality, and downstream decision dependencies can all affect the realized economic outcome. Additive stage-value models are useful in some settings, but they are special cases of a more general workflow-value function. This distinction is important because it prevents valuation from reducing agents to isolated components and instead places attention on the system in which they operate.

The coalition formulation shows why agent value is relational. The marginal contribution of an AI agent depends on the other human and artificial agents with which it operates, the tasks to which it is assigned, and the reliability losses or complementarities created by workflow interactions. The Shapley-value construction provides a principled attribution mechanism that allocates total surplus among participating agents and links valuation to accountability. The resulting Shapley pricing equilibrium offers a normative benchmark for comparing market prices with expected marginal economic contributions. In this sense, the same mathematical object that supports pricing also supports governance: determining which agents create value, which agents destroy value, and which agents warrant closer audit or supervision.

The security-operations case study illustrates how the framework can be applied to a hybrid human--AI workflow. It shows that agents with strong technical performance need not be economically optimal once deployment costs and failure losses are included, while agents with moderate standalone performance may create substantial value through complementarities. The example also demonstrates how attribution can inform accountability and pricing by decomposing total coalition surplus into agent-level marginal contributions.

Several directions remain open for future research. First, future work should develop empirical methods for estimating workflow-value functions from operational data, including settings in which value is only partially observed or is confounded by human supervision. Second, the reliability model can be extended to capture correlated failures, error propagation, adversarial manipulation, and recovery mechanisms rather than relying on conditional independence as a baseline. Third, dynamic models are needed for agents that learn, adapt, or degrade over time, as well as for workflows with stochastic task arrivals and changing organizational demand. Fourth, future work should incorporate endogenous supervision, in which humans and AI systems jointly decide when to verify, intervene, escalate, or delegate. Fifth, the pricing model can be extended to strategic markets with platform fees, bargaining power, asymmetric information, multi-homing, reputation, and competition among agent providers. Finally, the accountability dimension of Agentomics should be connected to regulatory requirements, audit trails, risk management, and organizational governance so that economic attribution can support responsible deployment of agentic AI systems in high-consequence domains.

\bibliographystyle{abbrv}	
\bibliography{refs.bib}
\end{document}